\documentclass{IEEEtran}

\usepackage{acronym}

\acrodef{BS}{base station}
\acrodefplural{BS}{base stations}

\acrodef{NOMA}{non-orthogonal multiple access}

\acrodef{OMA}{orthogonal multiple access}

\acrodef{RU}{resource unit}
\acrodefplural{RU}{resource units}

\acrodef{SIC}{successive interference cancellation}
\acrodef{SIF}{standard interference function}
\acrodefplural{SIF}{standard interference functions}
\acrodef{SINR}{signal-to-interference-and-noise ratio}

\acrodef{UE}{user equipment}
\acrodefplural{UE}{user equipments}

\usepackage{pstricks-add,pst-node,pst-tree}
\usepackage{epsfig}
\usepackage{pst-grad} 
\usepackage{pst-plot} 
\usepackage[space]{grffile} 
\usepackage{etoolbox} 
\makeatletter 
\patchcmd\Gread@eps{\@inputcheck#1 }{\@inputcheck"#1"\relax}{}{}
\makeatother

\usepackage{mathdots}
\usepackage[english]{babel}
\usepackage{relsize}
\usepackage{graphics,graphicx}
\usepackage{url}

\usepackage{booktabs}
\usepackage{multirow}
\usepackage{mparhack}
\usepackage{subfigure}
\usepackage{authblk}
\usepackage{amsthm}
\usepackage{mathrsfs}

\usepackage{array}
\usepackage{cite}
\usepackage{eqparbox}
\usepackage{mdwmath}

\usepackage{epsfig}
\usepackage{xcolor}

\usepackage{mathtools,cuted}
\usepackage{bm}
\usepackage{bbold}
\usepackage[all]{xy}
\usepackage{etoolbox}

\usepackage{clrscode3e}

\DeclareMathAlphabet{\mathantt}{OT1}{antt}{li}{it}
\DeclareMathAlphabet{\mathpzc}{OT1}{pzc}{m}{it}

\DeclarePairedDelimiter\norm{\lVert}{\rVert}%

\usepackage{accents}
\usepackage{multicol}

\usepackage{lipsum,color}

\usepackage{tcolorbox}

\usepackage{cite}

\usepackage{pgfplots}
\usepackage{tikz}
\usetikzlibrary{plotmarks}
\usetikzlibrary{patterns}
\usetikzlibrary{shapes.geometric}

\newtheorem{theorem}{Theorem}
\newtheorem{corollary}{Corollary}

\newtheorem{lemma}{Lemma}

\newtheorem{definition}{Definition}

\def\B{\mathcal{B}}

\def\I{\mathcal{I}}
\def\J{\mathcal{J}}
\def\T{\mathcal{T}}

\def\u{\mathpzc{u}}
\def\U{\mathcal{U}}

\DeclareMathOperator{\argmin}{\arg\min}

\renewcommand{\vec}[1]{\bm{#1}}

\usepackage{xparse}
\NewDocumentCommand{\overarrow}{O{=} O{\uparrow} m}{%
  \overset{\makebox[0pt]{\begin{tabular}{@{}c@{}}#3\\[0pt]\ensuremath{#2}\end{tabular}}}{#1}
}
\NewDocumentCommand{\underarrow}{O{=} O{\downarrow} m}{%
  \underset{\makebox[0pt]{\begin{tabular}{@{}c@{}}\ensuremath{#2}\\[0pt]#3\end{tabular}}}{#1}
}

\makeatletter
\def\blfootnote{\gdef\@thefnmark{}\@footnotetext}
\makeatother

\definecolor{mycolor}{rgb}{0.59215686274, 0, 0.01960784313}

\begin{document}

\title{A Note on Decoding Order in User Grouping and Power Optimization for Multi-Cell NOMA with Load Coupling}

\author{

    \IEEEauthorblockN{Lei~You and Di~Yuan}\\
    \IEEEauthorblockA{Department of Information Technology, Uppsala University, Sweden
    \\\{lei.you; di.yuan\}@it.uu.se
    }
%


}

\maketitle

\begin{abstract}
  In this technical note, we present a new theoretical result for
  multi-cell non-orthogonal multiple access (NOMA).
  For multi-cell scenarios, a so-called load-coupling model has been
  proposed earlier to characterize the presence of mutual interference for
  NOMA, and the optimization process relies on the use of fixed-point
  iterations \cite{8254922,8353846} across cells. One difficulty here
  is that the order of decoding for successive interference
  cancellation (SIC) in NOMA is generally not known a priori. This is
  because the decoding order in one cell depends on interference,
  which, in turn, is governed by resource usage in other cells,
  and vice versa. To achieve convergence, previous works have used
  workarounds that pose restrictions to NOMA, such that the SIC
  decoding order remains throughout the fixed-point iterations.  As a comment to
  \cite{8254922,8353846}, we derive and prove the following result:
  The convergence is guaranteed, even if the order changes over the
  iterations. The result not only waives the need of previous
  workarounds, but also implies that a wide class of 
  optimization problems for multi-cell NOMA is tractable, as long as
  that for single cell is.
\end{abstract}

\section{Introduction}
\label{sec:introduction}

\IEEEPARstart{N}{on-orthogonal} Multiple Access (NOMA) with \ac{SIC}
allows more than one user to share resource in the time-frequency
domain. With superposition coding, some users may perform SIC to remove (some of)
the intra-cell interference. The decoding needs to follow signal strength
so as to make SIC succeed. In the simplest case, the decoding order is
determined by the channel gains of users. NOMA in single-cell scenarios
has been widely addressed
\cite{6666209,7557079,8315498,7587811,7357604,7273963,7842433}. In
multi-cell NOMA, inter-cell interference has an influence on the
decoding order, which has to be accounted for
\cite{7676258,2016arXiv161101607S}. Besides, NOMA requires
\textit{user grouping} for resource sharing. The candidate options for
user grouping is exponential in the number of users. Third, power
allocation (a.k.a \textit{power split}) affects the performance.
Both user grouping and power split are intertwined with the decoding
order.

The decoding order in one cell depends on the interference from other
cells; the interference, in turn, depends on the allocated power and
time-frequency resources.  Increasing the resource consumption for
data transmission in one cell implies more interference to the users located in
other cells.  One approach for multi-cell optimization for \ac{OMA}
consists of fixed point iterations for the so-called \textit{load
  coupling} equation system
\cite{7959870,7132788,5198628,6204009,6732895,7585124,7880696,5450287,5489842,6292896,6363999,6479364,6747283,6924853,6815652,7151124,6887352,7480379,7273956,7332797,7962728,7744690,DBLP:journals/corr/abs-1710-09318}.
In every iteration, the resource allocation of one cells is computed,
with the allocation in other cells temporarily being fixed. At
convergence, an equilibrium with respect to resource allocation and
the resulting interference is obtained.  However, in NOMA, applying
the type of iterative method is challenging, since the decoding order
is not known beforehand, but subject to change during the iterative
process.  Some references
\cite{cui2018qoe,elbamby2017resource,baidas2018matching,kudathanthirige2019noma,zeng2019green}
do not explicitly address the interaction between decoding order and
interference. To the best of our knowledge,
\cite{7964738,8254922,8353846} consider this type of dependency
in multi-cell NOMA. In~\cite{7964738}, the
authors investigate multi-cell NOMA power control without user
grouping. References \cite{8254922,8353846} have used workarounds such
that some pre-computed decoding order remains NOMA-compliant in
optimization, which, on the other hand, poses limitations to the
applicable scenarios. To be specific, \cite{7964738} requires that
there is only one candidate group consisting of all users, and
\cite{8254922,8353846} require: 1) there are up to two users in each
group and 2) the candidate groups are selected such that the
NOMA-compliant decoding order can be pre-determined no matter the
interference. Without these conditions, neither the convergence nor
the optimality of their proposed algorithm is guaranteed.

This technical note serves as a comment to \cite{8254922,8353846},  though our main results are not necessarily bound to the specific system setups in \cite{8254922} and \cite{8353846}. The contributions are:

\begin{itemize}
\item We show a general conclusion with respect to the formulation of multi-cell NOMA optimization problems, such that the decoding order needs not to be explicitly ensured by constraints.
\item We use a so-called load-coupling model as an example, to showcase our conclusion.
The load coupling model has been widely adopted in OMA scenarios \cite{7959870,7132788,5198628,6204009,6732895,7585124,7880696,5450287,5489842,6292896,6363999,6479364,6747283,6924853,6815652,7151124,6887352,7480379,7273956,7332797,7962728,7744690,DBLP:journals/corr/abs-1710-09318} and has been extended to NOMA in \cite{8254922,8353846}. 
\item We formally proved that, the convergence and the optimality are guaranteed without imposing the limitations on the candidate user groups, even if the decoding order, due to variable inter-cell interference, changes from one iteration to the next in a fixed-point method. Furthermore, the result implies that a wide class of optimization problems for multi-cell NOMA is tractable as long as that for single cell is.
\end{itemize}

We clarify that in this technical note, the term ``\textit{correct
  decoding order}'', unless otherwise stated, always refers to the
NOMA-compliant decoding order. That is a receiver decodes the signals
successively in descending order of signal strengths.  This decoding
order in NOMA is based on the (valid) assumption that a user with
better channel condition (strong user) is able to decode the signal
transmitted to another user located in the same cell with worse
channel condition (weak user), as long as the weak user can decode the
signal of its own. In order words, the data rate encoded in the signal
to the weak user is permitted by the signal-to-interference--and-noise
ratio (SINR) of the weak user. Since the strong user can receive the
signal intended for the weak user better than the weak user itself,
the strong user can decode this signal as well, as the weak user can
do it (see, e.g., \cite{tse2005fundamentals}, Chapter~6.2.2, pp.~279
for further details).

Following up the above discussion, we remark that it is possible to
intentionally lower the rate transmitted to a user, in particular a
strong user, such that other, weak users can decode that signal and
perform interference cancellation. Such schemes have been studied in,
for example, \cite{AnChYu11,ZhTuTiBe14}. However, to our knowledge,
the literature of NOMA assume that the rate is set to what is maximumly
permitted by the SINR, which is the assumption in the current note as well.

\section{System Model}
\label{sec:sys_mod}

\subsection{Preliminaries}

Denote by $\I=\{1,2,\ldots,n\}$ the set of cells, and $\J$ the set of
\acp{UE}. We consider downlink and use $g_{ij}$ to denote the gain
from cell $i$ to \ac{UE} $j$.  For each cell $i$ ($i\in\I$), denote by
$\J_i$ the set of \acp{UE} located in the cell. Denote by $d_j$ the
data demand of \ac{UE} $j$ ($j\in\J$). Denote by $p_i$ the
transmission power of cell $i$ on each \ac{RU}. By using \ac{SIC},
multiple \acp{UE} of a cell can access one \ac{RU} simultaneously, with $p_i$
split among these \acp{UE}. We refer to the \acp{UE} sharing the same
\acp{RU} as a \textit{group}, and the process of selecting \acp{UE} to
form groups as \textit{user grouping}. We use $\u$ to refer to a
generic group. For any \ac{UE} $j\in\u$, denote by $q_{j\u}$ the
portion of power $p_i$ used for \ac{UE} $j$ on each \ac{RU} used by
 group $\u$. For cell $i$ ($i\in\I$), denote by $\U_i$ the set of
all groups of \acp{UE} in $\J_i$. In analogy with this, we use $\U_j$
to refer to the set of all groups that \ac{UE} $j$ belongs to. In
order to keep generality, we allow also singleton group $\u$. In this
case, the \ac{UE} in the group does not share \ac{RU} with others (and
hence no \ac{SIC}), i.e., the \ac{UE} is in
OMA. Besides, we allow one UE to belong to one or multiple groups such
that groups may have overlapping UEs. There is no limitation on the
number of UEs in one group. Both power split and user grouping are
subject to optimization.

In our derivations of the theoretical results, the channel information
is known. This is justified in research addressing the achievable
performance of a system. In our case, we target providing theoretical
results that are useful for computing the optimal performance of NOMA,
achieved with known channel information, and thereby assessing
accurately the potential benefit of NOMA in comparison to OMA.  In
practice, imperfect channel estimation will impact of the SIC
performance, even though SIC is more robust than parallel interference
cancellation (PIC) against error propagation~\cite{KoBoCa01}.
Moreover, we remark that there are several schemes for mitigating the
issue of propagation error for SIC, including soft-in soft-out (SISO)
decoding~\cite{KoBoCa01}, ordered SIC~\cite{Ki06}, and multi-feedback
and multi-branch SIC~\cite{Fala11,PeLaFa11}.

\subsection{NOMA with SIC}
\label{subsec:noma_sic}

In the following, we consider a generic UE $j$
located in cell $i$, i.e., $j \in \J_i$.
The \ac{SINR} of \ac{UE} $j$ in
group $\u$, denoted by $\gamma_{j\u}$, is given below.

\begin{equation}
\gamma_{j\u} = 	\frac{q_{j\u}g_{ij}}
{\underbrace{\sum_{h\in\u}
q_{h\u}g_{ij}\theta_{hj}}_{\textnormal{intra-cell}} 
+\underbrace{\sum_{k\in\I\backslash\{i\}}
   p_{k}g_{kj}\rho_{k}
}_{\textnormal{inter-cell}}
      +\sigma^2
},~{\Large \substack{ j\in\u \\ \\ \u\in\U_i}}
\label{eq:gamma}
\end{equation}

The denominator of $\gamma_{j\u}$ consists of three parts: intra-cell
interference, inter-cell interference, and noise power $\sigma^2$. The
optimization variables are the power split $q_{j\u}$, the
decoding order indicator $\theta_{hj}$ (discussed below), and the
cell-level resource allocation $\rho_{k}$ that is used as a scaling
parameter for the inter-cell interference (discussed below in
Section~\ref{subsec:load_coupling}). For each group $\u$, note that a
\ac{UE} $j$ of this group decodes the data of \acp{UE} with stronger
signal in $\u$, and the signals transmitted to other \acp{UE} of $\u$
constitute interference at $j$. We use $\theta_{hj}$ as a binary
indicator: $\theta_{hj}=1$ if and only if the signal transmission to
\ac{UE} $h$ is interference to \ac{UE} $j$, and $\theta_{hj}=0$ if and
only if \ac{UE} $j$ can decode the signal of UE $h$. As a convention,
$\theta_{hj}=0$ if $h=j$. We remark that $\theta_{hj}$ is subject to
the correct decoding order, which is determined by the channel condition
\cite{tse2005fundamentals}.  For any UE $j$, define

\begin{equation}
w_j = 	\left(\sum_{k\in\I\backslash\{i\}}
   p_{k}g_{kj}\rho_{k} + \sigma^2\middle)\right/g_{ij}.
\label{eq:w}
\end{equation}

For any user $j$, entity $w_j$ is the amount of inter-cell
interference and noise in relation to $j$'s own signal gain. Hence a
higher value of this entity means poorer channel condition.  For two
users $h$ and $j$ of the same cell, $h$ is the strong user if $w_h
\leq w_j$. Recall that, in NOMA, a strong user can decode the signal
of a weak user, as long as the weak user can decode the signal (of
which the rate is set to be that permitted by the weak user's SINR).
Thus, $\theta_{hj}=1$ if and only if $w_h\leq w_j$, i.e., UE $h$ has
better channel condition than UE $j$. We remark that in
\cite{tse2005fundamentals}, SIC is discussed for users of a single
cell. The result remains applicable here, because even if our scenario
consists in multiple cells, NOMA with SIC is used for each individual
cell, whereas inter-cell interference is treated in the same way as
noise. In addition, as remarked earlier, same as in \cite{tse2005fundamentals},
rate adaptation is not present in NOMA.

In case $w_h= w_j$, we assume there is a pre-defined convention to
break the tie. One possibility is to follow user index, e.g.,
$\theta_{hj}=1$ if $w_h= w_j$ and $h<j$.  As a result, for given inter-cell interference,
the NOMA decoding order is unique.

\subsection{Cell Load Coupling}
\label{subsec:load_coupling}

We define $c_{j\u}$ as the achievable capacity for UE $j$ in group $\u$ on one RU, namely,

\begin{equation}
\label{eq:cju}
c_{j\u} = \log(1+\gamma_{j\u}) = \log\left(1+\frac{q_{j\u}}{\sum_{h\in\u}q_{h\u}\theta_{hj}+w_j}\right).
\end{equation}

To ease the presentation of the mathematical proofs, in the theoretical derivations
we use the natural logarithm, and hence the information unit is
nat~\cite{ISO08} rather than bit. Thus, to be consistent, the data
demand $d_j, j \in \J$, are given in nats. 

The load coupling model defines $\rho_i$ to be the load of cell $i$,
which is the proportion of cell $i$'s time-frequency
resource that have been allocated for data transmission. 
Denote by $x_{\u}$ the proportion of \acp{RU} allocated to
group $\u$. We have $\rho_{i} = \sum_{\u\in\U_i}x_{\u}$. 
From an interference point of view, higher 
$\rho_i$ means that cell $i$ generates more interference to others,
hence cell load has been used as an interference scaling parameter
\cite{7959870,7132788,5198628,6204009,6732895,7585124,7880696,5450287,5489842,6292896,6363999,6479364,6747283,6924853,6815652,7151124,6887352,7480379,7273956,7332797,7962728,7744690,DBLP:journals/corr/abs-1710-09318,8254922,8353846},
see \eqref{eq:gamma}. We remark that, by~\eqref{eq:w}, for any UE $j$
in cell $i$, $w_j$ changes with other cells' resource
allocation $\rho_k$ ($k\in\I\backslash\{i\}$). Hence the decoding
order depends on the network-wide resource allocation.

We use $M$ and $B$ to represent the total number of \acp{RU} in one cell and the bandwidth of each \ac{RU}, respectively. To satisfy UE $j$'s demand $d_j$, we have
\begin{equation}
\sum_{\u\in\U_j}MBc_{j\u}x_{\u}\geq d_{j}, j\in\J,
\label{eq:dj}
\end{equation}
imposing that $d_j$ is satisfied by the sum of the demand delivered to UE $j$ over all groups in $\U_j$. In the discussion below, we use normalized $d_j$ such that the notation $M$ and $B$ are not necessary in our presentation.

\section{Problem Formulation}
\label{sec:problem}

\subsection{Mathematical Formulation}

We consider UE grouping and power split in NOMA, subject to inter-cell
interference that is modelled using the load-coupling model, as
outlined in the previous section. The cost function is defined with
respect to the amount of RUs needed to meet the user demand, i.e.,
cell load as defined earlier.  More specifically, we address power
split $\vec{q}$, group-level resource allocation $\vec{x}$, 
and cell-level resource allocation $\bm{\rho}$.
Moreover, as mentioned in Section~\ref{subsec:noma_sic}, $\bm{\theta}$
is the decoding order indicator that in turn depends on cell load.
The objective function $F$ is a generic cost function of the cell
loads (i.e. time-frequency resource usage of cells) and $F$ is
monotonically increasing in $\rho_1,\rho_2,\ldots,\rho_n$
element-wisely. The optimization formulation is given in~\eqref{eq:minF} below.
Note that in the formulation, the rate of a UE equals what is 
permitted by the SINR, where $\bm{\theta}$ is present for 
modeling intra-cell interference. In other words,
rate adaptation/optimization in SIC is not part of the model.

\begin{figure}[!ht]
\begin{subequations}
\begin{align}
& \min_{\substack{\vec{q},\vec{x},\bm{\rho},\vec{w}\geq\vec{0} \\ \bm{\theta}\in\{0,1\}}} F(\rho_1,\rho_2,\ldots,\rho_n) \\	
                     \text{s.t.} \quad & \sum_{\u\in\U_j}\log\left(1+\frac{q_{j\u}}{\sum_{h\in\u}q_{h\u}\theta_{hj}+w_j}\right) x_{\u} \geq d_j,~ j \in \J \label{eq:minF-d_j} \\
                     & w_j = \frac{\sum_{k\in\I\backslash\{i\}} p_{k}g_{kj}\rho_{k} + \sigma^2}{g_{ij}},~j\in\J_i,~i \in \I \label{eq:minF-w_j} \\
                     & \sum_{j\in\u}q_{j\u} \leq p_i,~\u\in\U_i,~i \in \I \label{eq:minF-p_i} \\
                     & \rho_i = \sum_{\u\in\U_i} x_{\u},~ i \in \I \label{eq:minF-rho_i} \\
                     &  w_j > w_h  \vee (w_j = w_h \land h < j) \Rightarrow  \theta_{hj}=1, \nonumber \\
& ~~~~~~~~~~~h \neq j, h, j \in \u, \u \in \U_i, i \in \I  \label{eq:minF-w_h-w_j} \\ 
                     & \theta_{hj} + \theta_{jh} = 1,~h\neq j,~h,j\in\u,\u \in \U_i, i \in \I \label{eq:minF-theta_hj-theta_jh} \\ 
                     & \theta_{hj}\in\{0,1\},~h\neq j,~h,j\in\u,\u \in \U_i, i \in \I  \label{eq:minF-theta_hj}
\end{align}
\label{eq:minF}
\end{subequations}
\end{figure}

The user demand constraints are~\eqref{eq:minF-d_j} and
\eqref{eq:minF-w_j}\footnote{With $d_j$ being normalized by $M\times
  B$ in this formulation, one can refer to \eqref{eq:gamma} and
  \eqref{eq:w} to verify that constraints~\eqref{eq:minF-d_j} along
  with~\eqref{eq:minF-w_j} are equivalent to~\eqref{eq:dj} in Section
  \ref{subsec:load_coupling}.}.  Constraints~\eqref{eq:minF-p_i}
impose the power limit, and equations ~\eqref{eq:minF-rho_i} define
cell load.
Constraints~\eqref{eq:minF-w_h-w_j}--\eqref{eq:minF-theta_hj} are for
the decoding order. Specifically, 
we have $\theta_{hj}=1$ if $w_h<w_j$ or if $h<j$ in case of a tie,
otherwise $\theta_{hj}=0$, following 
the rule of the correct decoding order in
Section~\ref{subsec:noma_sic}. 
Note that the decoding order depends on the cell level resource
allocation, i.e., $\rho_1,\rho_2,\ldots,\rho_n$. Moreover, as will be clear
later, stating \eqref{eq:minF-w_h-w_j} in logic form does not have any impact
on the theoretical results.

We remark that in reality there is an upper limit $\bar \rho$ for cell load, with the constraint
$\rho_i \leq \bar{\rho},~ i \in \I$. As will be shown later, this constraint can be accounted
for in post-processing.
Moreover, if necessary, user group selection constraints can be added to~\eqref{eq:minF}, and our conclusion 
in this paper still holds.

\subsection{Obstacles of Solving \eqref{eq:minF}}

We remark that~\eqref{eq:minF} is highly non-linear. In addition, a
major obstacle for some iterative algorithms for~\eqref{eq:minF} is
that, the variation of resource allocation in each iteration leads to
the change of decoding order for each group.  There is an algorithmic
framework derived in \cite{8254922,8353846}, which uses a top-down
paradigm (detailed in Section~\ref{subsec:solution}). The basic idea
is to break down~\eqref{eq:minF} into single cell level and then solve
the single-cell problems iteratively. For each iteration, there is an
inner loop over the cells. This inner loop can be performed in
parallel or sequentially. In the former case, the optimized resource
allocation of all cells serves as the input of the next iteration. In
the latter case, the optimized resource allocation of one cell updates
this cell's interference in the input, when the subproblem of another
cell is solved. Then, by fixed-point theory, the authors proved the
convergence of the algorithm, as well as the optimality of the
solution at convergence.

However, the algorithmic framework relies on two restrictions of
candidate groups, see\cite[Lemma 1]{8254922} and \cite[Lemma
1]{8353846}, respectively: Only those groups of which the decoding
orders can be pre-determined, are considered for optimization. The
other groups are eliminated from $\U$.  The limitation states that the
decoding orders of all the candidate groups must be independent of the
inter-cell interference such that the orders remain correct all the time, resulting
in sub-optimality.  To have a high probability of forming such groups,
\cite[Lemma 1]{8254922} and \cite[Lemma 1]{8353846} require each group
to consist of up to two UEs.

(\textit{Open Problem)} We remark that, if the restrictions are dropped, then in each iteration, the variation of cell loads may lead to the change of decoding order. In this case, neither convergence nor optimality is known.

\subsection{Solution of \cite{8254922,8353846} for \eqref{eq:minF}}
\label{subsec:solution}

By considering only groups for which the decoding order is independent of
interference, the variable $\bm{\theta}$ along with
\eqref{eq:minF-w_h-w_j}--\eqref{eq:minF-theta_hj} can be dropped
from~\eqref{eq:minF}, since $\bm{\theta}$ is pre-determined in
this special case.  The algorithmic framework in
\cite{8254922,8353846} is detailed as follows. Consider any cell $i$,
one can define the single-cell load minimization problem as a function
of the other cells' loads
$\bm{\rho}_{-i}=[\rho_1,\rho_2,\ldots,\rho_{i-1},\rho_{i+1},\ldots,\rho_n]$,
denoted by $\phi_i$:
\begin{equation}
\phi_i(\bm{\rho}_{-i}) = \min_{\vec{q},\vec{x},\vec{w}}\rho_i\text{ s.t. }\eqref{eq:minF-d_j}\text{--}\eqref{eq:minF-rho_i}~\text{of cell }i
\label{eq:fi}
\end{equation}

The authors proved that solving \eqref{eq:minF} amounts to obtaining the fixed point of
vector of functions $\bm{\phi} = [\phi_1, \phi_2, \dots \phi_n]$, that is, to solve
$\bm{\rho} = \bm{\phi}(\bm{\rho})$.  To be specific, the authors first
proved that $\phi_i$ ($i\in\I$) is a standard interference function
(SIF) \cite{414651}, of which the definition is given below.

\begin{definition}
Any function $\phi(\bm{\rho})$ that has the following two properties is an SIF, where $\bm{\rho}$ and $\bm{\rho}'$ are two arbitrary non-negative vectors with $\bm{\rho} \geq \bm{\rho}'$.
\begin{enumerate} 
\item (Scalability)
$\alpha \phi(\bm{\rho})>\phi(\alpha\bm{\rho})\text{~for any~} \alpha>1$.
\item (Monotonicity) $\phi(\bm{\rho})\geq \phi(\bm{\rho}')$.
\end{enumerate} 
\label{def:sif}
\end{definition}
Based on the fact that $\phi_i(\bm{\rho}_{-i})$ ($i\in\I$) is SIF, one
can obtain the unique fixed point $\bm{\rho}^{*}$ with
$\bm{\rho}^{*}=\bm{\phi}(\bm{\rho}^{*})$, by fixed-point iterations on
$\bm{\phi}$ \cite{414651}. Namely, for the iterative process
$\bm{\rho}^{(k+1)}=\bm{\phi}(\bm{\rho}^{(k)})$ ($k\geq 0$), we have
$\lim_{k\rightarrow\infty}\bm{\rho}^{(k)}=\bm{\rho}^{*}$, for
arbitrary non-negative starting point $\bm{\rho}^{(0)}$.  At
convergence of the fixed-point iterations, $\bm{\rho}^{*}$ along with
the other variables $\vec{q},\vec{x},\vec{w}$ that are obtained by
solving~\eqref{eq:fi} at $\bm{\rho}^{*}_{-i}$ for all $i\in\I$, is
optimal to~\eqref{eq:minF}.

\section{Results}
\label{sec:main_results}

This section derives our theoretical results, which give the answer to
the open problem in Section~\ref{sec:problem}. Our main conclusion is
that \cite[Lemma 1]{8254922} and \cite[Lemma 1]{8353846} can be
dropped, without loss of optimality or convergence of the proposed
solution methods. To show this, we first prove a general conclusion in
Section~\ref{subsec:rate-region} that is not tied to the load coupling
system. The conclusion states that, even if algebraically one allows
the capacity formula $c_{j\u}=\log(1+\gamma_{j\u})$ with ``decoding
orders'' in $\u$ to be all possible permutations of UEs of the group, the correct
decoding order leads to the largest $c_{j\u}$. Based on this, we prove
in Section~\ref{subsec:convergence} the convergence of the solution
methods. We then show the optimality after the convergence proof.

\subsection{Pseudo Rate Region}
\label{subsec:rate-region}

Consider \textit{rate region} at the RU level.  For a generic cell and
a user group $\u$ under consideration, we use $\bm{\theta}^{*}$ to
refer to the correct decoding order, i.e., the decoding order used by
NOMA for the users in $\u$. Note that $\bm{\theta}^{*}$ differs by user
group, and for a given group, $\bm{\theta}^{*}$ will change from one
iteration to another, when applying fixed-point iterations in solving
\eqref{eq:minF}.  However, within any iteration, when considering a
specific group, the amount of interference is given, and hence
$\bm{\theta}^{*}$ is easily determined as below, where entities $w_h$ and
$w_j$ contain the interference terms. For readability, we do not 
put cell, group, or iteration index on $\bm{\theta}^{*}$.

\[
\theta^{*}_{hj}=1\text{ iff } w_j > w_h  \text{ or } (w_j = w_h \text{ and } h < j), ~h \not=j, h, j \in \u.
\]

Suppose there are $K$ ($K\geq 2$) UEs multiplexed on an RU.  The UEs
are now indexed by following the correct decoding order as defined
above.  That is, UE $1$ decodes UEs $2,\ldots,K$.  UE $2$ has the
signal intended for UE $1$ as interference, and decodes UEs
$3,\ldots,K$, and so on.  The capacity of \ac{UE} $j$
($j=1,2,\ldots,K$), denoted by $c_{j}$, is
\[
c_{j}=\log\left(1+\frac{q_{j}}{\sum_{h=1}^{j-1}q_{h}+w_j}\right).
\]

Considering RU power limit $p$, the power split constraint reads 
\begin{equation}
\sum_{j=1}^{K}q_{j}\leq p.
\label{eq:p-max}
\end{equation}

The rates of UEs $1,2,\ldots K$ are as follows.

\begin{alignat*}{2}
& c_{1} = \log\left(1+ \frac{q_{1}}{w_1} \right)\\ 
& c_{2} = \log \left( 1+\frac{q_{2}}{q_{1}+w_2} \right)\\ 
& \vdots \\
& c_{K} = \log \left(1+\frac{q_{K}}{\sum_{h=1}^{K-1}q_{h}+w_K}\right)
\end{alignat*}\\
\label{eq:minF2}
\!\!For UE $1$, we have
\[
c_{1} = \log\left(1+\frac{q_{1}}{w_1}\right)	\Rightarrow q_{1} = w_{1}e^{c_{1}}-w_1
\]
For UE $2$, we have
\begin{multline*}
c_{2} = \log\left(1+\frac{q_{2}}{q_{1}+w_2}\right) \\ \Rightarrow q_{1}+q_{2} = w_1 e^{c_{1}+c_{2}} + (w_2-w_1)e^{c_{2}} - w_2	
\end{multline*}
By successively applying the same formula until the last user $K$, we obtain the equation below, where $w_0 = 0$.
\begin{equation}
R_{\bm{\theta}^*}(\vec{c})= \sum_{j\in\u} q_{j\u} = \sum_{t=1}^{K}\left(w_t-w_{t-1}\right)e^{\sum_{k=t}^{K}c_{t}} - w_{K}
\label{eq:A}
\end{equation}
where $\vec{c}=[c_{1},c_{2},\ldots,c_{K}]$, and $\bm{\theta}^{*}$ indicates the correct decoding order. Consequently, the power split constraint~\eqref{eq:p-max} is equivalent to the inequality below

\begin{equation}
R_{\bm{\theta}^*}(\vec{c})\leq p,
\label{eq:p-max2}
\end{equation}
where the power split variables $q_{1},q_{2},\ldots,q_{j}$ are replaced by variables $c_{1},c_{2},\ldots,c_{K}$ that represent the rates, respectively for UEs $1,2,\ldots,K$. 
Inequality \eqref{eq:p-max2} forms a bounded area and is the \textit{rate region} of all the $K$ UEs.

We remark that though the discussion above is based on applying the
successive rule on UEs by following their correct decoding order, the
rule is also applicable for the case that UEs are ordered arbitrarily.
Namely, for a group of UEs that are indexed in any given permutation
of the UEs, this successive rule also gives a formula with the same
form as \eqref{eq:A}. We introduce notations to represent this formula
in general. Define by $\T$ the index set of the $K!$ permutations of 
the $K$ UEs. For permutation $t \in \T$, $\pi_t(j)$ is the position
of UE $j$ in the permutation. 
Define $\B$ as a domain of $\bm{\theta}$, derived for all the permutations, i.e.,
\[
\B = \{\bm{\theta}: \theta_{hj} = 1 \text{~iff~} \pi_t(h) < \pi_t(j), t \in \T\}.
\]

We use $R_{\bm{\theta}}$ as a generic notation to represent \eqref{eq:A} defined for the order indicated by $\bm{\theta}$ ($\bm{\theta}\in\B$), so as to distinguish from the formula $R_{\bm{\theta}^{*}}$ that is specified for the correct decoding order. 
 
 We name the region defined by $R_{\bm{\theta}}(\bm{c})\leq p^{\max}$ with any $\bm{\theta}\in\B$ as \textit{pseudo rate region}. 
 \begin{equation}
R_{\bm{\theta}}(\bm{c})\leq p,~\bm{\theta}\in\B
 \label{eq:pseudo}
 \end{equation}
The reason for the name ``pseudo'' is because, with $\bm{\theta}$ ($\bm{\theta}\neq\bm{\theta}^{*}$), the \ac{SIC} may not be successfully performed for all UEs.

\begin{theorem}
Any pseudo rate region is a subset of the rate region of the correct decoding order. Namely, 
\[
\{\vec{c}:R_{\bm{\theta}}(\vec{c})\leq p\}\subseteq\{\vec{c}:R_{\bm{\theta}^{*}}(\vec{c})\leq p\}
\]
or equivalently,
\[
R_{\bm{\theta}^{*}}(\vec{c})\leq R_{\bm{\theta}}(\vec{c}),~\text{for~any~} \vec{c}\geq\vec{0}
\]
holds for any $\bm{\theta}\in\B$.
\label{thm:region}
\end{theorem}
\begin{proof}
Consider the pseudo rate region for $\bm{\theta}$ ($\bm{\theta}\in\B$), i.e. $R_{\bm{\theta}}(\vec{c})\leq p$. We index the UEs from $1$ to $K$ by following the order indicated by $\bm{\theta}$. We remark that if $\bm{\theta}$ is not the correct decoding order (i.e. $\bm{\theta}\neq\bm{\theta}^{*}$), then there must exist two UEs that are adjacent in the list, denoted by $\ell$ and $\ell+1$, such that $w_{\ell}>w_{\ell+1}$.
We swap the order of the two, and denote by $\bm{\theta}'$ the new decoding order. Below, we prove $R_{\bm{\theta}'}(\vec{c})\leq R_{\bm{\theta}}(\vec{c})$ for any non-negative $\vec{c}$.

To ease our representation, we define $w_{0}=0$ and $w_{K+1}=w_{K+2}$.  We also explicitly impose that for any summation notation ``$\sum_{t=a}^{b}$'' in our expression, if $b<a$, then this term in the sum equals zero.

For $\ell$ and $\ell+1$ ($\ell = 1,2,\ldots,K-1$), we have
\begin{align*}
R_{\bm{\theta}}(\vec{c}) & =  \sum_{t=1}^{\ell -1}\left(w_t-w_{t-1}\right)e^{\sum_{k=t}^{K}c_{k}} \\
                & \quad + (w_{\ell}-w_{\ell -1} )e^{\sum_{k=\ell }^{K}c_{k}} \\
                & \quad + (w_{\ell+1}-w_{\ell})e^{c_{\ell }+\sum_{k=\ell +2}^{K}c_{k}} \\
                & \quad + (w_{\ell+2}-w_{\ell+1})e^{\sum_{k=\ell +2}^{K}c_{k}} \\
                & \quad + \sum_{t=\ell+2}^{K}(w_{t+1}-w_{t})e^{\sum_{k=t+1}^{K}c_{k}}
\end{align*}
and
\begin{align*}
R_{\bm{\theta}'}(\vec{c}) & =  \sum_{t=1}^{\ell -1}\left(w_t-w_{t-1}\right)e^{\sum_{k=t}^{K}c_{k}} \\
                & \quad + (w_{\ell +1}-w_{\ell -1} )e^{\sum_{k=\ell }^{K}c_{k}} \\
                & \quad + (w_{\ell }-w_{\ell +1})e^{c_{\ell }+\sum_{k=\ell +2}^{K}c_{k}} \\
                & \quad + (w_{\ell+2}-w_{\ell})e^{\sum_{k=\ell +2}^{K}c_{k}} \\
                & \quad + \sum_{t=\ell+2}^{K}(w_{t+1}-w_{t})e^{\sum_{k=t+1}^{K}c_{k}}
\end{align*}

We remark that, the difference $R'(\vec{c})-R(\vec{c})$ makes the two summation terms in both the head and tail (if either exists) disappear. See~\eqref{eq:R'-A} below.
\begin{figure}[!ht]
\begin{align}
 & R_{\bm{\theta}'}(\vec{c}_{})-R_{\bm{\theta}}(\vec{c}) \nonumber\\
 =~ & (w_{\ell +1}-w_{\ell -1} )e^{\sum_{k=\ell }^{K}c_{k}}  + (w_{\ell }-w_{\ell +1})e^{c_{\ell }+\sum_{k=\ell +2}^{K}c_{k}} \nonumber\\
						   & +(w_{\ell+2}-w_{\ell})e^{\sum_{k=\ell +2}^{K}c_{k}} 
							 - (w_{\ell}-w_{\ell-1})e^{\sum_{k=\ell}^{K}c_{k}} \nonumber\\
						   &- (w_{\ell+1}-w_{\ell})e^{\sum_{k=\ell+1}^{K}c_{k}} 
						    - (w_{\ell+2}-w_{\ell+1})e^{\sum_{k=\ell+2}^{K}c_{k}} \nonumber\\
						 =~ & e^{\sum_{k=\ell}^{K}c_{k}}\big\{(w_{\ell +1}-w_{\ell -1} ) - (w_{\ell}-w_{\ell-1})\big\} \nonumber\\
						   & + e^{\sum_{k=\ell+2}^{K}c_{k}}\big\{(w_{\ell}-w_{\ell+1})e^{c_{\ell}} -(w_{\ell+1}-w_{\ell})e^{c_{\ell+1}}\big\} \nonumber\\
						   & +e^{\sum_{k=\ell+2}^{K}}\big\{(w_{\ell +2}-w_{\ell} )- (w_{\ell+2}-w_{\ell+1})\big\} \nonumber\\
						=~  & e^{\sum_{k=\ell}^{K}c_{k}}(w_{\ell +1}-w_{\ell }) +e^{\sum_{k=\ell+2}^{K}}(w_{\ell+1}-w_{\ell}) \nonumber\\
						   & +e^{\sum_{k=\ell+2}^{K}c_{k}}(w_{\ell+1}-w_{\ell})(-e^{c_{\ell}}-e^{c_{\ell+1}}) \nonumber\\
						=~   &  (w_{\ell+1}-w_{\ell})\bigg\{e^{\sum_{k=\ell}^{K}c_{k}}- e^{c_{\ell}+\sum_{k=\ell+2}^{K}c_{k}} \nonumber\\
						   & \quad\quad\quad\quad\quad\quad\quad - e^{\sum_{k=\ell+1}^{K}c_{k}}+e^{\sum_{k=\ell+2}^{K}c_{k}} \bigg\} \nonumber\\
						=~   &  (w_{\ell+1}-w_{\ell})e^{\sum_{k=\ell+2}^{K}c_{k}}\big\{e^{c_{\ell} + c_{\ell+1}} -e^{c_{\ell}} -e^{c_{\ell+1}} +1 \big\} \nonumber\\
						=~   &  (w_{\ell+1}-w_{\ell})e^{\sum_{k=\ell+2}^{K}c_{k}}
(e^{c_{\ell}}-1)(e^{c_{\ell+1}}-1)
\label{eq:R'-A}
\end{align}
\end{figure}

In the result of~\eqref{eq:R'-A}, because $w_{\ell}> w_{\ell+1}$ and $c_{\ell}\geq 0$ ($\ell=1,2,\ldots,K$), we conclude
\[
R_{\bm{\theta}'}(\vec{c})\leq R_{\bm{\theta}}(\vec{c}),~\text{for~any~} \vec{c}\geq\vec{0}.
\]

As a result,
\begin{equation}
\{\vec{c}:R_{\bm{\theta}}(\vec{c})\leq p\}\subseteq\{\vec{c}:R_{\bm{\theta}'}(\vec{c})\leq p\}
\label{eq:R'}
\end{equation}

The result in~\eqref{eq:R'} shows that, for two adjacent UEs $\ell$ and $\ell+1$ with $w_{\ell}> w_{\ell+1}$, 
swapping the order of the two UEs enlarges the pseudo rate region. 
We therefore conclude that the correct decoding order yields the largest rate region, namely, 
\[
\{\vec{c}:R_{\bm{\theta}}(\vec{c})\leq p\}\subseteq\{\vec{c}:R_{\bm{\theta}^{*}}(\vec{c})\leq p\}
\]
and, for any $\vec{c}\geq\vec{0}$,
\[
R_{\bm{\theta}^{*}}(\vec{c})\leq R_{\bm{\theta}}(\vec{c}).
\]
The above holds for any $\bm{\theta}\in\B$, and the theorem follows.
\end{proof}

\subsection{Convergence and Optimality of Fixed-Point Algorithm for Load Coupling with NOMA}
\label{subsec:convergence}

In this section, we investigate the convergence of the approach for
solving~\eqref{eq:minF} as outlined in Section~\ref{sec:problem},
without any restriction/limitation. We first re-define the single-cell problem
in~\eqref{eq:fi} in Section~\ref{sec:problem}, by taking into
consideration the
dependency between decoding order and interference.
\begin{equation}
f_i(\bm{\rho}_{-i}) = \min_{\vec{q},\vec{x},\vec{w},\bm{\theta}}\rho_i\text{ s.t. }\eqref{eq:minF-d_j}\text{--}\eqref{eq:minF-theta_hj}~\text{of cell }i
\label{eq:fi2}
\end{equation}

We remark that, though $\bm{\theta}$ is variable in \eqref{eq:fi2}, it
will induce the correct decoding order $\bm{\theta}^{*}$ by
constraints~\eqref{eq:minF-w_h-w_j}--\eqref{eq:minF-theta_hj}, because
here $\rho_{-i}$ and hence interference is known.  Therefore,
$\bm{\theta}$ is determined for any given $\bm{\rho}_{-i}$ in
\eqref{eq:fi2}. However, we do not eliminate the $\bm{\theta}$
variables from \eqref{eq:fi2}, because even though $\bm{\theta}^*$ is
directly induced by $\bm{\rho}_{-i}$, the latter is variable in the
fixed-point iterations, for which we will prove the convergence and
optimality.

We first prove Lemma~\ref{lma:c} below. In the
lemma, the definition of $c_{j\u}$ follows that in \eqref{eq:cju}, however
the dependency of $c_{j\u}$ as a function of $w_j$ is now made explicit.

\begin{lemma} 
Given non-negative $w_j$, the inequalities below hold for any $\alpha>1$.
\[
\frac{1}{\alpha}c_{j\u}(w_j)< c_{j\u}(\alpha w_j)
\]
\label{lma:c}
\end{lemma}
\begin{proof}
Since $1/c_{j\u}(w_j)$ is strictly concave in $\bm{\rho}_{-i}$, we have
\[
\frac{1}{c_{j\u}(\alpha w_j)}<\frac{\alpha}{c_{j\u}(w_j)} \Rightarrow \frac{1}{\alpha}c_{j\u}(w_j)< c_{j\u}(\alpha w_j)
\]
\end{proof}

We use $f_i(\bm{\rho}_{-i},\bm{\theta})$ to represent the optimization problem defined in~\eqref{eq:fi2} under any given $\bm{\theta}$ ($\bm{\theta}\in\B$). Mathematically:
\begin{equation}
f_i(\bm{\rho}_{-i},\bm{\theta}) = \min_{\vec{q},\vec{x},\vec{w}}\rho_i\text{ s.t. }\eqref{eq:minF-d_j}\text{--}\eqref{eq:minF-rho_i}~\text{of cell }i
\label{eq:fi_theta}
\end{equation}
We remark that by variable substitution as in Section~\ref{subsec:rate-region}, one has a reformulation of $f_i(\bm{\rho}_{-i},\bm{\theta})$, with $\vec{q}$ replaced by $\vec{c}$:
\begin{subequations}
\begin{align}
f_i(\bm{\rho}_{-i},\bm{\theta}) = & \min_{\vec{c},\vec{x},\vec{w}} \rho_i \\	
                     \text{s.t.} &\quad  R_{\bm{\theta}}(\vec{c},\vec{w})\leq p_i \label{eq:fi_theta_2-pi}\\
                     	 &\quad \eqref{eq:minF-w_j}\text{ and }\eqref{eq:minF-rho_i}\text{ of cell $i$} \nonumber
\end{align}
 \label{eq:fi_theta_2}
\end{subequations}

Note that the single-cell optimization problem above always has a
solution, because the problem is to minimize resource usage subject to
given demand and interference, and there is no 
upper limit imposed for resource usage. As a result, function $f_i(\bm{\rho}_{-i},\bm{\theta})$
is well defined.

\begin{lemma}
For any given $\bm{\theta}$ ($\bm{\theta}\in\B$), $f_i(\bm{\rho}_{-i},\bm{\theta})$ is an SIF of $\bm{\rho}_{-i}$ ($\bm{\rho}_{-i}\geq\bm{0}$).
\label{lma:sif}
\end{lemma}
\begin{proof}
(Monotonicity) Consider any $f_i(\bm{\rho}_{-i},\bm{\theta})$ ($\bm{\theta}\in\B$), and \eqref{eq:fi_theta}. For any $\bm{\rho}_{-i}$ and $\bm{\rho}'_{-i}$ with $\bm{\rho}'_{-i}\leq\bm{\rho}_{-i}$, we have $w_j(\bm{\rho}_{-i})\geq w_j(\bm{\rho}'_{-i})$ ($j\in\U_i$). Therefore $c_{j\u}(\bm{\rho}_{-i})\leq c_{j\u}(\bm{\rho}'_{-i})$. Replacing $c_{j\u}(\bm{\rho}_{-i})$ by $c_{j\u}(\bm{\rho}'_{-i})$ leads to a relaxation on the constraints~\eqref{eq:minF-d_j}, resulting in lower objective value. We then conclude $f_i(\bm{\rho}'_{-i},\bm{\theta})\leq f_i(\bm{\rho}_{-i},\bm{\theta})$ for any $\bm{\theta}\in\B$. 

(Scalability) Denote the value of $f_{i}(\bm{\rho}_{-i},\bm{\theta})$ by $\rho''_i$, i.e.  $\rho''_i=f_{i}(\bm{\rho}_{-i},\bm{\theta})$. Denote the optimal solution of $f_{i}(\bm{\rho}_{-i},\bm{\theta})$ by $\langle \vec{q}'', \vec{x}'', \vec{w}'' \rangle$. 
Under $\bm{\rho}_{-i}$, consider the following minimization problem. Denote its optimal objective value  by $z$.
\begin{subequations}
\begin{align}
  z = & \min_{\vec{q},\vec{w},\vec{x}} \rho_i \\	
                     \text{s.t.}   & \quad \frac{1}{\alpha}\sum_{\u\in\U_j} c_{j\u}(\vec{q},w_j)x_{\u}\geq  d_j,~j\in\J_i \label{eq:fib1-dj} \\
                     &\quad \eqref{eq:minF-w_j}\text{--}\eqref{eq:minF-rho_i}\text{ of cell $i$} \nonumber
\end{align}
 \label{eq:fib1}
\end{subequations}
It is straightforward to verify that $\langle \vec{q}'',\vec{w}'',\alpha\vec{x}'' \rangle$ is feasible to \eqref{eq:fib1}, with the objective value equaling $\alpha f_i(\bm{\rho}_{-i},\bm{\theta})$. We conclude that the optimum of \eqref{eq:fib1} is no higher than $\alpha f_i(\bm{\rho}_{-i},\bm{\theta})$. Namely, we have
\begin{equation}
z\leq \alpha f_i(\bm{\rho}_{-i},\bm{\theta}).
\label{eq:z1}
\end{equation}

For $f_{i}(\alpha\bm{\rho}_{-i},\bm{\theta})$, the corresponding formulation is as follows, where we remark that multiplying $\alpha$ on $\bm{\rho}_{-i}$ is equivalent to performing the multiplication on $w_j$ for all $j\in\J_i$.
\begin{subequations}
\begin{align}
 f_{i}(\alpha\bm{\rho}_{-i},\bm{\theta}) = & \min_{\vec{q},\vec{w},\vec{x}} \rho_i \\	
                     \text{s.t.} & \quad \sum_{\u\in\U_j} c_{j\u}(\vec{q},\alpha w_j)x_{\u}\geq d_j,~j\in\J_i \label{eq:fib2-dj} \\
                     &\quad \eqref{eq:minF-w_j}\text{--}\eqref{eq:minF-rho_i}\text{ of cell $i$} \nonumber
\end{align}
 \label{eq:fib2}
\end{subequations}

Note that \eqref{eq:fib2} differs from \eqref{eq:fib1} only in \eqref{eq:fib2-dj}, and \eqref{eq:fib1-dj} is equality at optimum. By Lemma~\ref{lma:c}, for any solution of \eqref{eq:fib1}, using it for~\eqref{eq:fib2} makes~\eqref{eq:fib2-dj} an inequality. Therefore, \eqref{eq:fib2} has a better optimum than \eqref{eq:fib1}, i.e.
\begin{equation}
f_i(\alpha\bm{\rho}_{-i},\bm{\theta})< z.	
\label{eq:z2}
\end{equation}
By \eqref{eq:z1} and \eqref{eq:z2}, $f_i(\alpha\bm{\rho}_{-i},\bm{\theta})<\alpha f_i(\bm{\rho}_{-i},\bm{\theta})$ holds. 
\end{proof}

\begin{lemma}
The NOMA decoding order $\bm{\theta}^{*}$ is optimal to $\min_{\bm{\theta}\in\mathcal{B}}f_i(\bm{\rho}_{-i},\bm{\theta})$ ($i\in\I$), i.e.,
\[
f_i(\bm{\rho}_{-i},\bm{\theta}^{*}) = \min_{\bm{\theta}\in\mathcal{B}}f_i(\bm{\rho}_{-i},\bm{\theta}),~i\in\I.
\]
\label{lma:fi_theta}	
\end{lemma}
\begin{proof}
Consider any cell $i$ ($i\in\I$) and any decoding order $\bm{\theta}$ other than the correct one. By Theorem~\ref{thm:region}, under fixed $\vec{c}$ and $\vec{w}$, replacing $R_{\bm{\theta}}(\vec{c},\vec{w})$ by $R_{\bm{\theta}^{*}}(\vec{c},\vec{w})$ makes the constraint~\eqref{eq:fi_theta_2-pi} remain satisfied (or relaxed if $R_{\bm{\theta}^{*}}(\vec{c},\vec{w})<R_{\bm{\theta}}(\vec{c},\vec{w})$), such that one will not get a worse objective value under the correct decoding order. Hence, we conclude that, at the optimum of  $\min_{\bm{\theta}\in\B} f_i(\bm{\rho}_{-i},\bm{\theta})$, $\bm{\theta}$ is (or can be replaced by) $\bm{\theta}^{*}$. Hence $\bm{\theta}^{*} \in  
\argmin_{\bm{\theta}\in\mathcal{B}}f_i(\bm{\rho}_{-i},\bm{\theta})$.
\end{proof}

By Lemma~\ref{lma:fi_theta}, the minimum of $f_i$ for any cell and user group is
achieved by using the NOMA decoding order $\bm{\theta}^{*}$.  In other
words, in any fixed-point iteration, using $\bm{\theta}^{*}$ that is
induced by the interference in that iteration, does not cause any loss
of optimality in the iteration's output.

We then prove Theorem~\ref{thm:sif} below. 

\begin{theorem}
The function $f_i(\bm{\rho}_{-i})$ ($i\in\I$) in~\eqref{eq:fi2} is SIF. 
\label{thm:sif}
\end{theorem}
\begin{proof}

  First, note that
  $f_i(\bm{\rho}_{-i})=f_i(\bm{\rho}_{-i},\bm{\theta}^{*})$, by the
  definitions of the two functions $f_i(\bm{\rho}_{-i})$ and
  $f_i(\bm{\rho}_{-i},\bm{\theta})$, and $\bm{\theta}^{*}$. Second, by
  Lemma~\ref{lma:fi_theta}, we have
  $f_i(\bm{\rho}_{-i},\bm{\theta}^{*}) =
  \min_{\bm{\theta}\in\mathcal{B}}f_i(\bm{\rho}_{-i},\bm{\theta})$.
  Therefore, we conclude
  $f_i(\bm{\rho}_{-i})=\min_{\bm{\theta}\in\mathcal{B}}f_i(\bm{\rho}_{-i},\bm{\theta})$. The theorem follows then from the fact that the minimum of finitely many SIFs is also an SIF~\cite{414651}.
\end{proof}

The following corollary shows an algorithmic framework for optimally
solving problem~\eqref{eq:minF}. Briefly, one only needs to apply
fixed-point iterations on all $f_i$ ($i\in\I$) to reach 
optimality. Given cell loads $\bm{\rho}_{-i}$, evaluating
$f_i(\bm{\rho}_{-i})$ ($i\in\I$) submits to a single-cell load
minimization problem.

\begin{corollary}
Assume that problem~\eqref{eq:minF} has a solution, then
the iterations $\bm{\rho}^{(k+1)}=\vec{f}(\bm{\rho}^{(k)})$, with arbitrary starting point $\bm{\rho}^{(0)}$ ($\bm{\rho}^{(0)}\geq \vec{0}$), converge to a unique fixed-point $\bm{\rho}^{*}$, such that $\bm{\rho}^{*}=\vec{f}(\bm{\rho}^{*})$.
Let $\vec{q}^{*},\vec{x}^{*},\vec{w}^{*},\bm{\theta}^{*}$ be the solution obtained by solving the problems $f_i(\bm{\rho}^{*}_{-i})$ for all $i\in\I$. Then for problem~\eqref{eq:minF}, we have
\begin{enumerate}
\item The optimal solution is $\vec{q}^{*},\vec{x}^{*},\vec{w}^{*},\bm{\theta}^{*}$.
\item The optimal objective value is $F(\rho_1^{*},\rho_2^{*},\ldots,\rho_n^{*})$.
\end{enumerate}
\label{thm:opt}
\end{corollary}
The proof of Corollary~\ref{thm:opt} can be straightforwardly derived based on Theorem~\ref{thm:sif}. One can refer to \cite[Theorem 3]{8254922} or \cite[Theorem 6]{8353846} for more details. 

By Corollary~\ref{thm:opt}, \eqref{eq:minF} has a unique solution
$\bm{\rho}^{*} = [\rho^*_1, \dots, \rho^*_n]$.  Therefore, with
post-processing, one can determine if $\bm{\rho}^{*}$ is within the
resource limit, i.e., if $\rho_i^* \leq {\bar \rho},~i \in \I$;
the violation of the constraint for any cell means the problem is infeasible.

\section{Discussion}
\label{sec:discussion}
This section discusses the potential application of our derived results in Section~\ref{sec:main_results}
along three dimensions: \textit{problem formulation, tractability, and optimality}.

\subsection{Decoding Order Constraints}
\label{subsec:order?}

It is worth noting that constraints~\eqref{eq:minF-w_h-w_j} are redundant for~\eqref{eq:minF}. Consider the formulation below.
\[
\min_{\vec{q},\vec{x},\bm{\rho},\vec{w},\bm{\theta}\geq\vec{0}} F(\rho_1,\rho_2,\ldots,\rho_n)\text{ s.t. } \eqref{eq:minF-d_j}\text{--}\eqref{eq:minF-rho_i},~\eqref{eq:minF-theta_hj-theta_jh},\text{ and }\eqref{eq:minF-theta_hj}
\]
Theorem~\ref{thm:region} along with the analysis in Section~\ref{sec:main_results} indeed reveals that at its optimum, $\bm{\theta}=\bm{\theta}^{*}$. Namely, for any $\theta_{hj}$ at the optimum, if $\theta_{hj}=1$, we must have $w_h\leq w_j$, meaning that $\theta_{hj}$ satisfies $\theta_{hj} \geq \min\{1,w_j - w_h\}$. Hence the non-linear constraints \eqref{eq:minF-w_h-w_j} are redundant and can be removed from the formulation. 

\subsection{Tractability of~\eqref{eq:minF}}
\label{subsec:tractability?}

We remark that even though~\eqref{eq:minF} is for  multi-cell scenarios, the difficulty indeed lies in its corresponding single-cell load minimization problems, i.e., \eqref{eq:fi2}. By Lemma~\ref{lma:fi_theta}  we reach the optimum at $\bm{\theta}^{*}$.
 Once~$f_i(\bm{\rho}_{-i},\bm{\theta}^{*})$ can be solved to optimality, then as pointed out by Corollary~\ref{thm:opt}, the optimum of~\eqref{eq:minF} can be straightforwardly obtained. There are some special cases of \eqref{eq:fi_theta} that submit to a polynomial-time solution, briefly discussed below. 

 If the power allocation $\vec{q}$ is fixed, then the single-cell load
 minimization problem \eqref{eq:fi} is a linear programming problem in
 $\vec{x}$ and $\vec{w}$ \cite{8254922}. We remark that
 \eqref{eq:fi_theta} can be reformulated to \eqref{eq:fi_theta_2} by
 using the successive rule in Section~\ref{subsec:rate-region}. As the
 second case, if the demand on each user group is given, then 
 variable $\vec{x}$ can be eliminated, and \eqref{eq:fi_theta_2}
 is a convex programming formulation\footnote{We remark that the
   convexity of constraints~\eqref{eq:fi_theta_2-pi} holds if
   $\bm{\theta}$ is set to be the correct decoding order). By
   Lemma~\ref{lma:fi_theta}, we know that this is the only case that
   needs to be taken into account.} of $\vec{c}$ and
 $\vec{w}$\cite{zhu2017optimal}.

 Consider another case where the number of UEs in each group is no
 more than two (i.e. $|\u|\leq 2$ for all $\u\in\U$, $i \in \I$), and there is no
 overlapping UE for any two selected groups. Then \eqref{eq:minF} can
 be solved optimally within polynomial time by combinatorial
 optimization \cite{8353846}. The basic idea is to bring the
 single-cell load optimization down to the group level, and then prove
 that the optimal group selection amounts to solving a maximum
 matching problem.

For the general case of~\eqref{eq:fi_theta}, whether it is tractable or not, remains open, suggesting future research to be done along this direction.

\subsection{A Comment on \cite{8254922,8353846}}
\label{subsec:optimality?}

In \cite{8254922,8353846}, the proposed solution has guaranteed
convergence if some restrictions, \cite[Lemma~1]{8254922} and
\cite[Lemma~1]{8353846}, are imposed. The results derived in this
technical note provide a complementary theoretical insight, namely,
the restrictions can be dropped without any loss of optimality or
convergence.

\section{Numerical Results}
\label{sec:numerical}

In this section, we provide performance comparison of OMA and NOMA,
as well as numerical validation of our theoretical findings. We use a
cellular network of 19 cells, with wrap-around for eliminating edge
effects. The simulation settings are given in
\tablename~\ref{tab:sim}.

\begin{table}[!ht]
\centering
\caption{Simulation Parameters.}
\begin{tabular}{ll}
\toprule
\textbf{Parameter} & \textbf{Value} \\
Cell radius & $500$ m\\
Carrier frequency & $2$ GHz \\
Total bandwidth & $20$ MHz\\
Cell load limit $\bar{\rho}$ & $1.0$ \\
Path loss model & COST-231-HATA \\
Shadowing (Log-normal) & $6$ dB standard deviation\\
Fading & Rayleigh flat fading \\
Noise power spectral density & $-173$ dBm/Hz \\
RB power $p_i$ ($i\in\I$) & $800$ mW \\
Convergence tolerance ($\epsilon$) & $10^{-4}$ \\
\bottomrule
\end{tabular}
\label{tab:sim}	
\end{table}

The sum load minimization problem, i.e.,
$F(\bm{\rho})=\sum_{i\in\I}\rho_i$ in formulation~\eqref{eq:minF},
can be used to examine the maximum throughput supported, for uniform
user demand. A given demand level (of all users) can be supported, if
and only if load minimization leads to load levels that are all within
the load limit.  Hence the maximum throughput can be computed by solving
\eqref{eq:minF} repeatedly with a bi-section search on the demand level.  For
this experiment, there are 30 UEs randomly distributed inside each
cell. Each user group $\u$ contains two UEs.
  We remark that both OMA and NOMA are solved
to optimality.  In OMA, the optimum is obtained by using the method in
\cite{6204009}.  In NOMA, we use the algorithm proposed in
\cite{8353846}. However, the constraint imposed by \cite[Lemma
1]{8353846} on the candidate user groups is dropped, and, by our
theoretical results, there is no loss of global optimality.
  
\pgfplotsset{compat=1.11,
        /pgfplots/ybar legend/.style={
        /pgfplots/legend image code/.code={%
        \draw[##1,/tikz/.cd,bar width=3pt,yshift=-0.2em,bar shift=0pt]
                plot coordinates {(0cm,0.8em)};},
},
}
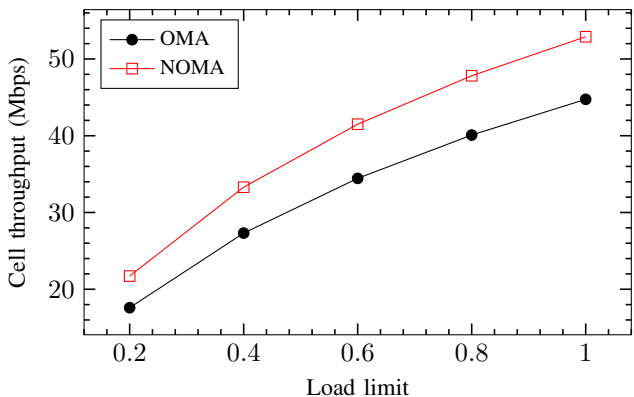
\begin{figure}[!h] 
\centering 
\begin{tikzpicture}
\begin{axis}[
	xlabel={Load limit},
	ylabel={Cell throughput (Mbps)},
	label style = {font=\fontsize{9pt}{10pt}\selectfont},
	legend cell align={left},
	legend pos = north west,
	legend style = {font=\fontsize{8pt}{10pt}\selectfont},
	axis background/.style={fill=white},
minor x tick num=4,
minor y tick num=4,
major tick length=0.15cm,
minor tick length=0.075cm,
tick style={semithick,color=black},
	height=0.667\linewidth,
	width=\linewidth,
]
	
\addplot [ mark=*, color=black] coordinates {
	(0.2, 17.5968)
	(0.4, 27.3131)
	(0.6, 34.4448)
	(0.8, 40.0885)
	(1.0, 44.7426)
};

\addplot [ mark= square, color=red] coordinates{
	(0.2, 21.7188)
	(0.4, 33.2947)
	(0.6, 41.5001)
	(0.8, 47.8128)
	(1.0, 52.8954)
};

\legend{OMA, NOMA}
\end{axis}
\end{tikzpicture}
\caption{Cell throughput (in Mbps)
in function of load limit.}
\label{fig:load} 
\end{figure}

\figurename~\ref{fig:load} shows the performance 
in terms of the throughput (in Mbps) per cell, for various
levels of the cell load limit. We remark that in the derivations
in the previous sections, the data is in nats, merely to simplify
the mathematical expressions. Here, for convenience, we show the results
in bits. Note that all users of a cell, including those located
at cell edge, are provided with the same data rate.
From \figurename~\ref{fig:load}, NOMA leads to considerably higher throughput
than OMA, while consuming the same amount of resource. The performance
improvement is about 20\% or higher. Moreover, even though the absolute difference
increases with the cell load limit, the relative difference is in fact higher
when the load limit is low, that is, NOMA offers more performance
gain when the use of resource is more constrained. 

\pgfplotsset{compat=1.11,
        /pgfplots/ybar legend/.style={
        /pgfplots/legend image code/.code={%
        \draw[##1,/tikz/.cd,bar width=3pt,yshift=-0.2em,bar shift=0pt]
                plot coordinates {(0cm,0.8em)};},
},
}
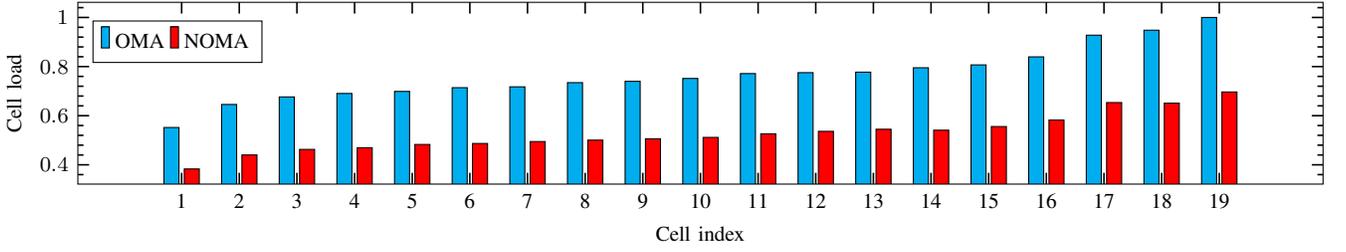
\begin{figure*}[!ht]
\centering
\begin{tikzpicture}
\begin{axis}[
	font=\fontsize{8pt}{10pt}\selectfont,
    ybar,
    legend style={at={(0.08,0.9)},
      anchor=north,legend columns=-1},
    ylabel={Cell load},
    xlabel={Cell index},
    symbolic x coords={1,2,3,4,5,6,7,8,9,10,11,12,13,14,15,16,17,18,19},
    xtick=data,
    	height=4cm,
	width=\textwidth,
	bar width = 0.2cm,
	xtick align=inside,
	minor y tick num=4,
	major tick length=0.15cm,
	minor tick length=0.075cm,
	tick style={semithick,color=black},
    ]
\addplot [draw=black, fill=cyan] coordinates { (1,0.5518861459993285) (2,0.6456801704701509) (3,0.6765038792178337) (4,0.6904469122376065) (5,0.69883609666146) (6,0.7140220134439574) (7,0.7175419491828879) (8,0.734539176494171) (9,0.7400260501227063) (10,0.7516544758359655) (11,0.7714676520929664) (12,0.7752305327497773) (13,0.7773449197257659) (14,0.7951420393176793) (15,0.806836912680391) (16,0.8395434052240811) (17,0.9277948087169523) (18,0.948015122217504) (19,0.9999976880691546)};
\addplot [draw=black, fill=red] coordinates {(1,0.3831101100077303) (2,0.4400796300584208) (3,0.46269308927566954) (4,0.4696107789002239) (5,0.4827656334496884) (6,0.4868277144342972) (7,0.4947032922736718) (8,0.5005532965317314) (9,0.5060115235843202) (10,0.511560077301295) (11,0.5262871092337478) (12,0.5362236826548274) (13,0.5447604243315015) (14,0.540933229652313) (15,0.5557959259517387) (16,0.5822758280551295) (17,0.6534936481549145) (18,0.651100515117767) (19,0.6966427658475751)};

\legend{OMA, NOMA}
\end{axis}
\end{tikzpicture}
\caption{Cell load levels sorted in ascending order for NOMA, with load limit 1.0. }
\label{fig:cell-load}
\end{figure*}

It should be pointed out that not all cells are able to utilize 100\%
of the resource simultaneously to maximize the throughput.  To clarify
further this aspect, in \figurename~\ref{fig:cell-load} we show the
individual cell load levels at the maximum achievable OMA throughput,
with load limit being equal to one, i.e., the load levels of OMA and
NOMA corresponding the right-end of the two curves in
\figurename~\ref{fig:load}. Indeed, for OMA, one can observe that one
cell has reached the limit, whereas the other cells still have spare
resource. However, these cells are not able to consume the spare
resource, because doing so would lead to higher interference, as
captured by the load-coupling model, overloading the cell that
currently is fully loaded. While delivering the same throughput as
OMA, by NOMA the load is less than 70\% in all cells. As none of the
cells has exhausted its resource, additional throughput can be
offered.


\pgfplotsset{compat=1.11,
        /pgfplots/ybar legend/.style={
        /pgfplots/legend image code/.code={%
        \draw[##1,/tikz/.cd,bar width=3pt,yshift=-0.2em,bar shift=0pt]
                plot coordinates {(0cm,0.8em)};},
},
}
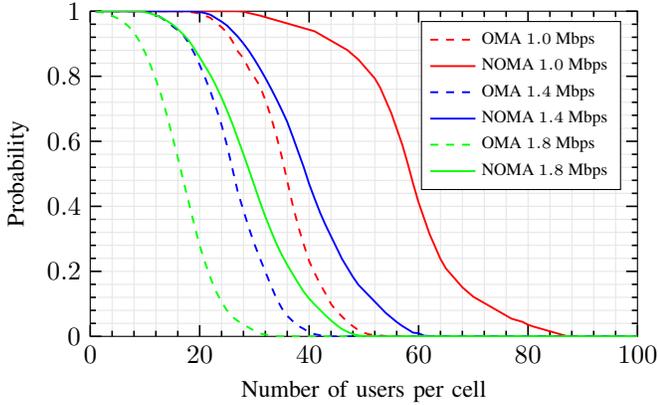
\begin{figure}[!ht] 
\centering 
\begin{tikzpicture}
\begin{axis}[
	xlabel={Number of users per cell},
	ylabel={Probability},
	label style = {font=\fontsize{9pt}{10pt}\selectfont},
	legend cell align={left},
	legend pos = north east,
	legend style = {font=\fontsize{6.5pt}{7pt}\selectfont},
	axis background/.style={fill=white},
	grid=both,
	grid style={gray!20},
    xmin=0,xmax=100,
    ymin=0,ymax=1,
minor x tick num=4,
minor y tick num=4,
major tick length=0.15cm,
minor tick length=0.075cm,
tick style={semithick,color=black},
	height=0.667\linewidth,
	width=\linewidth,
]
	
\addplot [line width=0.25mm, color=red, dashed] coordinates {
( 1.0 , 1 )
( 2.0 , 1 )
( 3.0 , 1 )
( 4.0 , 1 )
( 5.0 , 1 )
( 6.0 , 1 )
( 7.0 , 1 )
( 8.0 , 1 )
( 9.0 , 1 )
( 10.0 , 1 )
( 11.0 , 1 )
( 12.0 , 1 )
( 13.0 , 1 )
( 14.0 , 1 )
( 15.0 , 1 )
( 16.0 , 1 )
( 17.0 , 0.99936 )
( 18.0 , 0.99769 )
( 19.0 , 0.99602 )
( 20.0 , 0.99268 )
( 21.0 , 0.98599 )
( 22.0 , 0.9793 )
( 23.0 , 0.96425 )
( 24.0 , 0.95087 )
( 25.0 , 0.93415 )
( 26.0 , 0.90739 )
( 27.0 , 0.87562 )
( 28.0 , 0.85562 )
( 29.0 , 0.82218 )
( 30.0 , 0.79706 )
( 31.0 , 0.77027 )
( 32.0 , 0.72679 )
( 33.0 , 0.67662 )
( 34.0 , 0.61809 )
( 35.0 , 0.55287 )
( 36.0 , 0.47762 )
( 37.0 , 0.40404 )
( 38.0 , 0.33882 )
( 39.0 , 0.28364 )
( 40.0 , 0.2318 )
( 41.0 , 0.19501 )
( 42.0 , 0.16157 )
( 43.0 , 0.13147 )
( 44.0 , 0.0997 )
( 45.0 , 0.07796 )
( 46.0 , 0.06124 )
( 47.0 , 0.04452 )
( 48.0 , 0.03114 )
( 49.0 , 0.01943 )
( 50.0 , 0.01107 )
( 51.0 , 0.00605 )
( 52.0 , 0.00438 )
( 53.0 , 0.00271 )
( 54.0 , 0.00103 )
( 55.0 , 0 )
( 56.0 , 0 )
( 57.0 , 0 )
( 58.0 , 0 )
( 59.0 , 0 )
( 60.0 , 0 )
( 61.0 , 0 )
( 62.0 , 0 )
( 63.0 , 0 )
( 64.0 , 0 )
( 65.0 , 0 )
( 66.0 , 0 )
( 67.0 , 0 )
( 68.0 , 0 )
( 69.0 , 0 )
( 70.0 , 0 )
( 71.0 , 0 )
( 72.0 , 0 )
( 73.0 , 0 )
( 74.0 , 0 )
( 75.0 , 0 )
( 76.0 , 0 )
( 77.0 , 0 )
( 78.0 , 0 )
( 79.0 , 0 )
( 80.0 , 0 )
( 81.0 , 0 )
( 82.0 , 0 )
( 83.0 , 0 )
( 84.0 , 0 )
( 85.0 , 0 )
( 86.0 , 0 )
( 87.0 , 0 )
( 88.0 , 0 )
( 89.0 , 0 )
( 90.0 , 0 )
( 91.0 , 0 )
( 92.0 , 0 )
( 93.0 , 0 )
( 94.0 , 0 )
( 95.0 , 0 )
( 96.0 , 0 )
( 97.0 , 0 )
( 98.0 , 0 )
( 99.0 , 0 )
( 100.0 , 0 )
};

\addplot [line width=0.25mm, color=red] coordinates {
( 1.0 , 1 )
( 2.0 , 1 )
( 3.0 , 1 )
( 4.0 , 1 )
( 5.0 , 1 )
( 6.0 , 1 )
( 7.0 , 1 )
( 8.0 , 1 )
( 9.0 , 1 )
( 10.0 , 1 )
( 11.0 , 1 )
( 12.0 , 1 )
( 13.0 , 1 )
( 14.0 , 1 )
( 15.0 , 1 )
( 16.0 , 1 )
( 17.0 , 1 )
( 18.0 , 1 )
( 19.0 , 1 )
( 20.0 , 1 )
( 21.0 , 1 )
( 22.0 , 1 )
( 23.0 , 1 )
( 24.0 , 1 )
( 25.0 , 1 )
( 26.0 , 1 )
( 27.0 , 0.9998 )
( 28.0 , 0.9988 )
( 29.0 , 0.99428 )
( 30.0 , 0.99045 )
( 31.0 , 0.98662 )
( 32.0 , 0.98179 )
( 33.0 , 0.97696 )
( 34.0 , 0.97213 )
( 35.0 , 0.9673 )
( 36.0 , 0.96247 )
( 37.0 , 0.95764 )
( 38.0 , 0.95281 )
( 39.0 , 0.94798 )
( 40.0 , 0.94315 )
( 41.0 , 0.93832 )
( 42.0 , 0.92866 )
( 43.0 , 0.919 )
( 44.0 , 0.90934 )
( 45.0 , 0.89968 )
( 46.0 , 0.89002 )
( 47.0 , 0.88036 )
( 48.0 , 0.86587 )
( 49.0 , 0.85138 )
( 50.0 , 0.83206 )
( 51.0 , 0.81274 )
( 52.0 , 0.79342 )
( 53.0 , 0.76443 )
( 54.0 , 0.72578 )
( 55.0 , 0.68713 )
( 56.0 , 0.63882 )
( 57.0 , 0.59051 )
( 58.0 , 0.53254 )
( 59.0 , 0.46974 )
( 60.0 , 0.41177 )
( 61.0 , 0.36346 )
( 62.0 , 0.31515 )
( 63.0 , 0.2765 )
( 64.0 , 0.23785 )
( 65.0 , 0.20886 )
( 66.0 , 0.18954 )
( 67.0 , 0.17022 )
( 68.0 , 0.1509 )
( 69.0 , 0.13641 )
( 70.0 , 0.12192 )
( 71.0 , 0.11226 )
( 72.0 , 0.1026 )
( 73.0 , 0.09294 )
( 74.0 , 0.08328 )
( 75.0 , 0.07362 )
( 76.0 , 0.06396 )
( 77.0 , 0.0543 )
( 78.0 , 0.04947 )
( 79.0 , 0.04464 )
( 80.0 , 0.03581 )
( 81.0 , 0.03015 )
( 82.0 , 0.02532 )
( 83.0 , 0.02049 )
( 84.0 , 0.01566 )
( 85.0 , 0.01083 )
( 86.0 , 0.006 )
( 87.0 , 0.00111 )
( 88.0 , 0.00111 )
( 89.0 , 0.00111 )
( 90.0 , 0 )
( 91.0 , 0 )
( 92.0 , 0 )
( 93.0 , 0 )
( 94.0 , 0 )
( 95.0 , 0 )
( 96.0 , 0 )
( 97.0 , 0 )
( 98.0 , 0 )
( 99.0 , 0 )
( 100.0 , 0 )
};

\addplot [line width=0.25mm, color=blue, dashed] coordinates {
( 1.0 , 1 )
( 2.0 , 1 )
( 3.0 , 1 )
( 4.0 , 1 )
( 5.0 , 1 )
( 6.0 , 1 )
( 7.0 , 1 )
( 8.0 , 1 )
( 9.0 , 1 )
( 10.0 , 0.99969 )
( 11.0 , 0.99607 )
( 12.0 , 0.99064 )
( 13.0 , 0.97977 )
( 14.0 , 0.96709 )
( 15.0 , 0.95441 )
( 16.0 , 0.94173 )
( 17.0 , 0.92724 )
( 18.0 , 0.90188 )
( 19.0 , 0.87108 )
( 20.0 , 0.83485 )
( 21.0 , 0.79499 )
( 22.0 , 0.75151 )
( 23.0 , 0.7026 )
( 24.0 , 0.65006 )
( 25.0 , 0.58665 )
( 26.0 , 0.51419 )
( 27.0 , 0.44173 )
( 28.0 , 0.38738 )
( 29.0 , 0.33303 )
( 30.0 , 0.28412 )
( 31.0 , 0.24064 )
( 32.0 , 0.19897 )
( 33.0 , 0.1573 )
( 34.0 , 0.11744 )
( 35.0 , 0.08664 )
( 36.0 , 0.06309 )
( 37.0 , 0.04679 )
( 38.0 , 0.0323 )
( 39.0 , 0.02324 )
( 40.0 , 0.01237 )
( 41.0 , 0.00694 )
( 42.0 , 0.00332 )
( 43.0 , 0.0015 )
( 44.0 , 0 )
( 45.0 , 0 )
( 46.0 , 0 )
( 47.0 , 0 )
( 48.0 , 0 )
( 49.0 , 0 )
( 50.0 , 0 )
( 51.0 , 0 )
( 52.0 , 0 )
( 53.0 , 0 )
( 54.0 , 0 )
( 55.0 , 0 )
( 56.0 , 0 )
( 57.0 , 0 )
( 58.0 , 0 )
( 59.0 , 0 )
( 60.0 , 0 )
( 61.0 , 0 )
( 62.0 , 0 )
( 63.0 , 0 )
( 64.0 , 0 )
( 65.0 , 0 )
( 66.0 , 0 )
( 67.0 , 0 )
( 68.0 , 0 )
( 69.0 , 0 )
( 70.0 , 0 )
( 71.0 , 0 )
( 72.0 , 0 )
( 73.0 , 0 )
( 74.0 , 0 )
( 75.0 , 0 )
( 76.0 , 0 )
( 77.0 , 0 )
( 78.0 , 0 )
( 79.0 , 0 )
( 80.0 , 0 )
( 81.0 , 0 )
( 82.0 , 0 )
( 83.0 , 0 )
( 84.0 , 0 )
( 85.0 , 0 )
( 86.0 , 0 )
( 87.0 , 0 )
( 88.0 , 0 )
( 89.0 , 0 )
( 90.0 , 0 )
( 91.0 , 0 )
( 92.0 , 0 )
( 93.0 , 0 )
( 94.0 , 0 )
( 95.0 , 0 )
( 96.0 , 0 )
( 97.0 , 0 )
( 98.0 , 0 )
( 99.0 , 0 )
( 100.0 , 0 )
};

\addplot [line width=0.25mm, color=blue] coordinates {
( 1.0 , 1 )
( 2.0 , 1 )
( 3.0 , 1 )
( 4.0 , 1 )
( 5.0 , 1 )
( 6.0 , 1 )
( 7.0 , 1 )
( 8.0 , 1 )
( 9.0 , 1 )
( 10.0 , 1 )
( 11.0 , 1 )
( 12.0 , 1 )
( 13.0 , 1 )
( 14.0 , 1 )
( 15.0 , 1 )
( 16.0 , 1 )
( 17.0 , 1 )
( 18.0 , 1 )
( 19.0 , 0.99928 )
( 20.0 , 0.99755 )
( 21.0 , 0.99408 )
( 22.0 , 0.98888 )
( 23.0 , 0.97848 )
( 24.0 , 0.96808 )
( 25.0 , 0.95595 )
( 26.0 , 0.93862 )
( 27.0 , 0.92129 )
( 28.0 , 0.90049 )
( 29.0 , 0.87796 )
( 30.0 , 0.85196 )
( 31.0 , 0.82423 )
( 32.0 , 0.79477 )
( 33.0 , 0.76357 )
( 34.0 , 0.73064 )
( 35.0 , 0.69598 )
( 36.0 , 0.65958 )
( 37.0 , 0.61452 )
( 38.0 , 0.56773 )
( 39.0 , 0.52094 )
( 40.0 , 0.46895 )
( 41.0 , 0.42389 )
( 42.0 , 0.38056 )
( 43.0 , 0.34243 )
( 44.0 , 0.30777 )
( 45.0 , 0.27311 )
( 46.0 , 0.23845 )
( 47.0 , 0.21072 )
( 48.0 , 0.18299 )
( 49.0 , 0.15873 )
( 50.0 , 0.1414 )
( 51.0 , 0.12407 )
( 52.0 , 0.10674 )
( 53.0 , 0.08941 )
( 54.0 , 0.07208 )
( 55.0 , 0.05822 )
( 56.0 , 0.04436 )
( 57.0 , 0.03223 )
( 58.0 , 0.02183 )
( 59.0 , 0.01144 )
( 60.0 , 0.00971 )
( 61.0 , 0.002 )
( 62.0 , 0.00101 )
( 63.0 , 0 )
( 64.0 , 0 )
( 65.0 , 0 )
( 66.0 , 0 )
( 67.0 , 0 )
( 68.0 , 0 )
( 69.0 , 0 )
( 70.0 , 0 )
( 71.0 , 0 )
( 72.0 , 0 )
( 73.0 , 0 )
( 74.0 , 0 )
( 75.0 , 0 )
( 76.0 , 0 )
( 77.0 , 0 )
( 78.0 , 0 )
( 79.0 , 0 )
( 80.0 , 0 )
( 81.0 , 0 )
( 82.0 , 0 )
( 83.0 , 0 )
( 84.0 , 0 )
( 85.0 , 0 )
( 86.0 , 0 )
( 87.0 , 0 )
( 88.0 , 0 )
( 89.0 , 0 )
( 90.0 , 0 )
( 91.0 , 0 )
( 92.0 , 0 )
( 93.0 , 0 )
( 94.0 , 0 )
( 95.0 , 0 )
( 96.0 , 0 )
( 97.0 , 0 )
( 98.0 , 0 )
( 99.0 , 0 )
( 100.0 , 0 )
};

\addplot [line width=0.25mm, color=green, dashed] coordinates {
( 1.0 , 0.99964 )
( 2.0 , 0.99693 )
( 3.0 , 0.99286 )
( 4.0 , 0.98608 )
( 5.0 , 0.97794 )
( 6.0 , 0.96709 )
( 7.0 , 0.95352 )
( 8.0 , 0.93452 )
( 9.0 , 0.90738 )
( 10.0 , 0.87346 )
( 11.0 , 0.83411 )
( 12.0 , 0.78933 )
( 13.0 , 0.73913 )
( 14.0 , 0.68079 )
( 15.0 , 0.61566 )
( 16.0 , 0.55053 )
( 17.0 , 0.47997 )
( 18.0 , 0.41213 )
( 19.0 , 0.34429 )
( 20.0 , 0.27916 )
( 21.0 , 0.22624 )
( 22.0 , 0.17875 )
( 23.0 , 0.13804 )
( 24.0 , 0.10548 )
( 25.0 , 0.0797 )
( 26.0 , 0.05935 )
( 27.0 , 0.04714 )
( 28.0 , 0.03629 )
( 29.0 , 0.02679 )
( 30.0 , 0.01729 )
( 31.0 , 0.00915 )
( 32.0 , 0.00508 )
( 33.0 , 0.00237 )
( 34.0 , 1.00E-03 )
( 35.0 , 0 )
( 36.0 , 0 )
( 37.0 , 0 )
( 38.0 , 0 )
( 39.0 , 0 )
( 40.0 , 0 )
( 41.0 , 0 )
( 42.0 , 0 )
( 43.0 , 0 )
( 44.0 , 0 )
( 45.0 , 0 )
( 46.0 , 0 )
( 47.0 , 0 )
( 48.0 , 0 )
( 49.0 , 0 )
( 50.0 , 0 )
( 51.0 , 0 )
( 52.0 , 0 )
( 53.0 , 0 )
( 54.0 , 0 )
( 55.0 , 0 )
( 56.0 , 0 )
( 57.0 , 0 )
( 58.0 , 0 )
( 59.0 , 0 )
( 60.0 , 0 )
( 61.0 , 0 )
( 62.0 , 0 )
( 63.0 , 0 )
( 64.0 , 0 )
( 65.0 , 0 )
( 66.0 , 0 )
( 67.0 , 0 )
( 68.0 , 0 )
( 69.0 , 0 )
( 70.0 , 0 )
( 71.0 , 0 )
( 72.0 , 0 )
( 73.0 , 0 )
( 74.0 , 0 )
( 75.0 , 0 )
( 76.0 , 0 )
( 77.0 , 0 )
( 78.0 , 0 )
( 79.0 , 0 )
( 80.0 , 0 )
( 81.0 , 0 )
( 82.0 , 0 )
( 83.0 , 0 )
( 84.0 , 0 )
( 85.0 , 0 )
( 86.0 , 0 )
( 87.0 , 0 )
( 88.0 , 0 )
( 89.0 , 0 )
( 90.0 , 0 )
( 91.0 , 0 )
( 92.0 , 0 )
( 93.0 , 0 )
( 94.0 , 0 )
( 95.0 , 0 )
( 96.0 , 0 )
( 97.0 , 0 )
( 98.0 , 0 )
( 99.0 , 0 )
( 100.0 , 0 )
};

\addplot [line width=0.25mm, color=green] coordinates {
( 1.0 , 1 )
( 2.0 , 1 )
( 3.0 , 1 )
( 4.0 , 1 )
( 5.0 , 1 )
( 6.0 , 1 )
( 7.0 , 1 )
( 8.0 , 1 )
( 9.0 , 1 )
( 10.0 , 0.9972 )
( 11.0 , 0.993 )
( 12.0 , 0.9874 )
( 13.0 , 0.9804 )
( 14.0 , 0.9706 )
( 15.0 , 0.9594 )
( 16.0 , 0.94539 )
( 17.0 , 0.92858 )
( 18.0 , 0.91037 )
( 19.0 , 0.88516 )
( 20.0 , 0.85855 )
( 21.0 , 0.83054 )
( 22.0 , 0.80253 )
( 23.0 , 0.77032 )
( 24.0 , 0.73391 )
( 25.0 , 0.69469 )
( 26.0 , 0.65267 )
( 27.0 , 0.60645 )
( 28.0 , 0.56023 )
( 29.0 , 0.51261 )
( 30.0 , 0.46359 )
( 31.0 , 0.41457 )
( 32.0 , 0.36835 )
( 33.0 , 0.32633 )
( 34.0 , 0.28571 )
( 35.0 , 0.2507 )
( 36.0 , 0.21989 )
( 37.0 , 0.19188 )
( 38.0 , 0.16387 )
( 39.0 , 0.13726 )
( 40.0 , 0.11625 )
( 41.0 , 0.09804 )
( 42.0 , 0.08123 )
( 43.0 , 0.06442 )
( 44.0 , 0.04901 )
( 45.0 , 0.035 )
( 46.0 , 0.02239 )
( 47.0 , 0.01399 )
( 48.0 , 0.00699 )
( 49.0 , 0.00279 )
( 50.0 , 0.0014 )
( 51.0 , 0 )
( 52.0 , 0 )
( 53.0 , 0 )
( 54.0 , 0 )
( 55.0 , 0 )
( 56.0 , 0 )
( 57.0 , 0 )
( 58.0 , 0 )
( 59.0 , 0 )
( 60.0 , 0 )
( 61.0 , 0 )
( 62.0 , 0 )
( 63.0 , 0 )
( 64.0 , 0 )
( 65.0 , 0 )
( 66.0 , 0 )
( 67.0 , 0 )
( 68.0 , 0 )
( 69.0 , 0 )
( 70.0 , 0 )
( 71.0 , 0 )
( 72.0 , 0 )
( 73.0 , 0 )
( 74.0 , 0 )
( 75.0 , 0 )
( 76.0 , 0 )
( 77.0 , 0 )
( 78.0 , 0 )
( 79.0 , 0 )
( 80.0 , 0 )
( 81.0 , 0 )
( 82.0 , 0 )
( 83.0 , 0 )
( 84.0 , 0 )
( 85.0 , 0 )
( 86.0 , 0 )
( 87.0 , 0 )
( 88.0 , 0 )
( 89.0 , 0 )
( 90.0 , 0 )
( 91.0 , 0 )
( 92.0 , 0 )
( 93.0 , 0 )
( 94.0 , 0 )
( 95.0 , 0 )
( 96.0 , 0 )
( 97.0 , 0 )
( 98.0 , 0 )
( 99.0 , 0 )
( 100.0 , 0 )
};

\legend{OMA $1.0$ Mbps, NOMA $1.0$ Mbps, OMA $1.4$ Mbps, NOMA $1.4$ Mbps, OMA $1.8$ Mbps, NOMA $1.8$ Mbps}
\end{axis}
\end{tikzpicture}
\caption{This figure illustrates the cumulative density function of having all users supported in every cell. The demand values in the legend are for each user.}
\label{fig:cdf} 
\end{figure}

As our next of numerical study, we consider the number of users
that can be supported per cell by OMA and NOMA, for various demand
levels. To this end, for each demand level, we run 1,000 realizations
for different numbers of users per cell, and record whether or not OMA
and NOMA can simultaneously support all users. The results are then
collected to generate the cumulative density function (CDF) shown
in \figurename~\ref{fig:cdf}. For the lowest demand used, OMA can
support approximately 30 users in every cell with 90\% probability.
For NOMA, the corresponding number is close to 45 users, representing
an increase of 50\%. Similar amount of improvement is observed for the
higher demand levels. Moreover, for all the demand levels, when the
number of users per cell has reached the level such that the
probability of supporting all of them by OMA is virtually zero, NOMA still
has approximately 40\% probability of supporting this number and
higher. Hence the performance enhancement, in terms of the number of
users that can be accommodated by the network, is significant.

Next, we make a comparison of
spectral efficiency (in bps/Hz) of OMA and NOMA, with respect to the
proportion of cell-edge users. Cell-edge users are those such that the
distance to the home base station is at least 80\% of the cell radius.
The results are provided in \figurename~\ref{fig:spec}. Again, by
our theoretical findings, the numerical results here represent the
achievable performance, and hence the comparison is accurate.  From the
figure, the impact of  cell-edge users on performance is very apparent.
Specifically, when there is no cell-edge user at all, 
the spectral efficiency is 30\% higher for NOMA.  The
spectral efficiency drops then quickly with respect to the proportion
of cell-edge users, and the improvement offered by NOMA has a
diminishing trend\footnote{Even though in numbers, NOMA still delivered more
than 10\% higher efficiency with 20\% cell-edge users.}.
That is, much less improvement can be expected from NOMA
for scenarios where many users 
are at cell edge.

\pgfplotsset{compat=1.11,
        /pgfplots/ybar legend/.style={
        /pgfplots/legend image code/.code={%
        \draw[##1,/tikz/.cd,bar width=3pt,yshift=-0.2em,bar shift=0pt]
                plot coordinates {(0cm,0.8em)};},
},
}
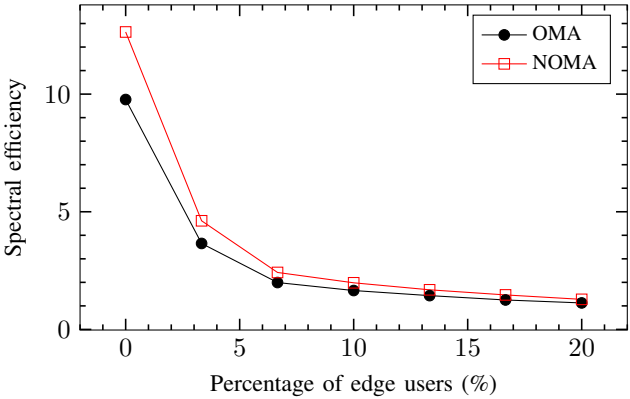
\begin{figure}[!ht] 
\centering 
\begin{tikzpicture}
\begin{axis}[
	xlabel={Percentage of edge users (\%)},
	ylabel={Spectral efficiency},
	label style = {font=\fontsize{9pt}{10pt}\selectfont},
	legend cell align={left},
	legend pos = north east,
	legend style = {font=\fontsize{8pt}{10pt}\selectfont},
	axis background/.style={fill=white},
minor x tick num=4,
minor y tick num=4,
major tick length=0.15cm,
minor tick length=0.075cm,
tick style={semithick,color=black},
	height=0.667\linewidth,
	width=\linewidth,
]
	
\addplot [ mark=*, color=black] coordinates {
	(0, 9.76979)
	(3.33, 3.65346)
	(6.66, 1.99219)
	(10.00, 1.6525)
	(13.33, 1.43616)
	(16.66, 1.25081)
	(20.00, 1.12557)
};

\addplot [ mark= square, color=red] coordinates{
	(0, 12.6405)
	(3.33, 4.61649)
	(6.66, 2.42073)
	(10.00, 1.98105)
	(13.33, 1.68339)
	(16.66, 1.46594)
	(20.00, 1.27779)
};

\legend{OMA, NOMA}
\end{axis}
\end{tikzpicture}
\caption{This figure illustrates the spectral efficiency as function of the proportion of of cell edge users.}
\label{fig:spec} 
\end{figure}

In \figurename~\ref{fig:convergence}, we show the convergence as well
as the convergence rate of fixed-point iterations of function
$\bm{f}(\bm{\rho})$. The theoretical guarantee of convergence, as stated earlier in
this note, is indeed observed. Moreover, high accuracy of the network
load can be reached after very few iterations on $\bm{f}(\bm{\rho})$,
and the algorithm converges slightly faster for higher user
demand.


\pgfplotsset{compat=1.12,
        /pgfplots/ybar legend/.style={
        /pgfplots/legend image code/.code={%
        \draw[##1,/tikz/.cd,bar width=3pt,yshift=-0.2em,bar shift=0pt]
                plot coordinates {(0cm,0.8em)};},
},
}
\begin{figure}[!ht]  
\begin{tikzpicture}
\begin{axis}[
	ymode = log,
	xlabel={Iteration $k$},
	ylabel={$\norm{\bm{\rho}^{(k)}-\bm{\rho}^{(k-1)}}_{\infty}$},
	label style = {font=\fontsize{9pt}{10pt}\selectfont},
	legend cell align={left},
	legend pos = north east,
	legend style = {font=\fontsize{8pt}{10pt}\selectfont},
	axis background/.style={fill=white},
minor x tick num=0,
minor y tick num=4,
major tick length=0.15cm,
minor tick length=0.075cm,
xtick = {2,3,4,...,20},
tick style={semithick,color=black},
	height=0.667\linewidth,
	width=0.97\linewidth,
	xmin=2,
	xmax=15,
]
	
\addplot [color=black,  mark=*] coordinates {
	(2, 0.0577461)
	(3, 0.00203188)
	(4, 0.0000852924)
	(5, 3.65109*10^-6)
	(6, 1.36205*10^-7)
	(7, 6.07099*10^-9)
	(8, 2.07329*10^-10)
	(9, 8.0064*10^-12)
	(10, 3.09183*10^-13)
	(11, 1.19397*10^-14)
	(12, 4.61075*10^-16)
	(13, 1.78053*10^-17)
	(14, 6.87588*10^-19)
	(15, 2.65526*10^-20)
	(16, 1.02538*10^-21)
	(17, 3.95971*10^-23)
	(18, 1.52912*10^-24)
	(19, 5.90499*10^-26)
	(20, 2.28033*10^-27)
};

\addplot [color=blue,  mark=square] coordinates {
	(2, 0.0496822)
	(3, 0.00239899)
	(4, 0.000124095)
	(5, 6.43253*10^-6)
	(6, 3.34073*10^-7)
	(7, 1.73539*10^-8)
	(8, 5.88846*10^-10)
	(9, 2.60267*10^-11)
	(10, 1.15037*10^-12)
	(11, 5.08458*10^-14)
	(12, 2.24736*10^-15)
	(13, 9.93322*10^-17)
	(14, 4.39044*10^-18)
	(15, 1.94055*10^-19)
	(16, 8.57716*10^-21)
	(17, 3.79106*10^-22)
	(18, 1.67563*10^-23)
	(19, 7.40622*10^-25)
	(20, 3.27351*10^-26)
};

\addplot [color=red,  mark=o] coordinates {
	(2, 0.0241175)
	(3, 0.00151652)
	(4, 0.000109303)
	(5, 8.59927*10^-6)
	(6, 5.46573*10^-7)
	(7, 2.7344*10^-8)
	(8, 1.41131*10^-9)
	(9, 7.62465*10^-11)
	(10, 5.11924*10^-12)
	(11, 2.22543*10^-13)
	(12, 1.2023*10^-14)
	(13, 6.49545*10^-16)
	(14, 3.50919*10^-17)
	(15, 1.89585*10^-18)
	(16, 1.02424*10^-19)
	(17, 5.53348*10^-21)
	(18, 2.98948*10^-22)
	(19, 1.61508*10^-23)
	(20, 8.72551*10^-25)
};
\legend{$d=1.5$ Mbps, $d=0.75$ Mbps, $d=0.15$ Mbps}
\end{axis}
\end{tikzpicture}
\caption{This figure shows the norm $\norm{\cdot}_{\infty}$ in function of iteration $k$ ($k\geq 2$), 
under the uniform demand settings of $d=1.5,~0.75,~\text{and}~0.15$~Mbps for each user, respectively, with 30 
users per cell.}
\label{fig:convergence}
\end{figure}
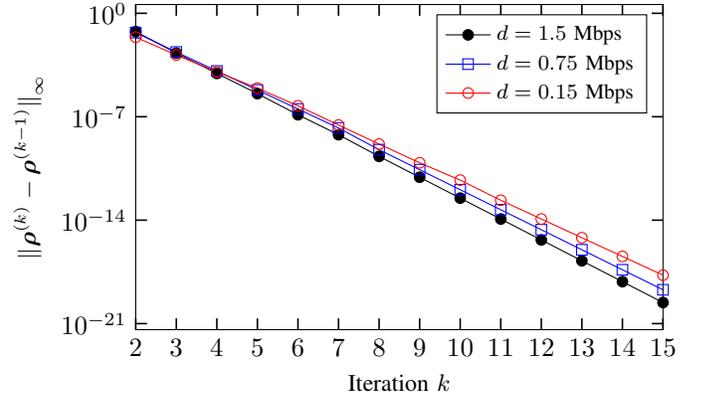

To summarize, NOMA provides considerable throughput gain over OMA, the
underlying reason is that, for the same demand level, NOMA requires
significantly less resource usage in all cells, and the spare resource
can be utilized for more throughput. For the same reason, NOMA
exhibits signficant improvement in the number of users that can be
supported.  However, the amount of performance improvement to expect
is tied to the proportion of cell-edge users. The simulation study
also confirms the correctness of our theoretical analysis of
convergence guarantee of the fixed-point algorithm.

\section{Conclusion}
\label{sec:conclusion}
This technical note has addressed the convergence and optimality of an
algorithmic framework for solving a class of optimization
problems in multi-cell NOMA networks. The note proved that the correct
decoding order corresponds to the largest region. Then, results for
convergence and optimality, with variable decoding order in the
iterative process, are formally established. The note has also
discussed the tractability of multi-cell optimization with
load coupling, and reveals that solving the single-cell problem is the
key in terms of tractability.

We remark that for contraction mapping, the convergence rate of
fixed-point iterations is linear, see~\cite{lemmens_nussbaum_2012}.
However for general SIFs the convergence can be sub-linear.  Hence,
an interesting topic of further study is to examine conditions under
which the SIF for NOMA optimization falls within the domain of
contraction mapping, to shed light on the convergence rate in addition
to the convergence results proved in the current paper.

As was discussed in Section~\ref{sec:introduction}, in NOMA it is
assumed that the data rate used for a user equals what is permitted by the
SINR. Suppose the rate is also subject to selection, such that it can
be set to be lower than the SINR-rate, in order to enable interference
cancellation. Although such a scheme is not part of the original NOMA,
combining this aspect with user grouping and power split of NOMA opens
up a new research line for future work.

A further line of our future research consists in accounting for the
impact of imperfect channel estimation on NOMA performance, as well as
means mitigating the problem of propagation error. The study amounts to
incorporating these aspects in the system model, and investigating the
resulting multi-cell NOMA optimization problem, for which the results
of the current paper will be used for benchmarking purposes.

\section*{Acknowledgment}

We would like to thank the reviewers and the editor for the valuable
comments that have enabled us to improve the paper, as well as
the inspirations thanks to the comments on future work (e.g.,
combining rate selection with NOMA is inspired
by a comment of Reviewer~3.)

\bibliographystyle{IEEEtran}
\bibliography{ref}

\end{document}